\documentclass[11pt]{article}

\usepackage{amsmath}
\usepackage{amssymb}
\usepackage{amsthm}
\usepackage{graphicx}
\usepackage{natbib}
\usepackage[left=25.4mm,right=25.4mm,top=30mm,bottom=25.4mm,a4paper]{geometry}
\usepackage{subfigure}
\theoremstyle{definition}
\newtheorem{theorem}{Theorem}
\newtheorem{lemma}[theorem]{Lemma}
\newtheorem{corollary}[theorem]{Corollary}

\newtheorem{remark}[theorem]{Remark}

\newtheorem{definition}{Definition}
\newtheorem{example}{Example}

\newcommand{\E}{\mathbb E}
\newcommand{\e}{\mathrm e}
\newcommand{\Q}{\mathbb Q}
\newcommand{\D}{\mathrm{d}}
\newcommand{\F}{\mathcal F}

\begin{document}
\title{Recursive formula for arithmetic Asian option prices}
\author{Kyungsub Lee\footnote{Department of Mathematical Sciences, KAIST, Daejeon, 305-701, Korea, Email: klee@euclid.kaist.ac.kr, Tel: +82-42-350-5709, Fax: +82-42-350-2710}
}
\date{}
\maketitle

\begin{abstract}
We derive a recursive formula for arithmetic Asian option prices with finite observation times in semimartingale models.
The method is based on the relationship between the risk-neutral expectation of the quadratic variation of the return process and European option prices.
The computation of arithmetic Asian option prices is straightforward whenever European option prices are available.
Applications with numerical results under the Black-Scholes framework and the exponential L\'evy model are proposed.
\end{abstract}

\section{Introduction}\label{Sect:intro}
An arithmetic Asian option is a financial derivative whose payoff depends on the arithmetic average of the underlying asset prices with pre-determined observation times.
Asian options are more robust with respect to manipulations of the underlying asset near maturity, in contrast to standard European option.
However, there is no known closed form formula for arithmetic Asian option prices.
In this paper we derive a recursive formula for arithmetic Asian option prices and show that these prices are consistent with European option prices based on the quadratic variation method.

The quadratic variation of return process, defined by the limit of the sum of squared returns,
plays a crucial role in financial analysis, as it is used for measuring volatility and high-moment risk of return distributions.
One of the most important properties of the quadratic variation process of the return of a financial asset is that the risk-neutral expectation of the variation over a fixed time period is synthesized by European options prices.
Thus, the risk-neutral expectation of the quadratic variation of the return process is computed by an integration formula whose integrand is  composed of weighted European option prices.

As an application, \cite{CarrWu} computed the option implied variance risk premium by comparing the realized variation of the return and risk-neutral expectation of the variation.
Also, \cite{ChoeLee} established a new method of measuring high moments of return distributions under both risk-neutral and physical probabilities.
They show that the risk-neutral expectation of a certain stochastic integration with respect to the quadratic variation process of return is represented by an integration formula whose integrand is composed of European option prices.
One of the interesting applications of this result is arithmetic Asian option pricing.
We demonstrate how to derive the risk-neutral expectation of an arithmetic Asian option price, also showing that the price is consistent with European option prices.

In previous studies that derive the price of an arithmetic Asian option, there are approximate solutions using geometric Asian option prices \citep{Ruttiens, Ton1992179}, methods using approximate distributions \citep{Levy1992, Turnbull1991, Bouaziz1994}, those based on the fast Fourier transforttm \citep{Carverhill} and Monte Carlo simulation techniques \citep{Kemna}.
\cite{Vecer2001} and \cite{Vecer2002} explained how to price both continuously and discretely monitored Asian options in the geometric Brownian motion model based on a PDE approach.
\cite{VecerXu} derived an integro-differential equation for an Asian option price when the underlying price process is assumed to follow the exponential L\'evy process.

The works of \cite{Vecer2001} and \cite{Vecer2002} were extended by \cite{Fouque} to the case of the stochastic volatility model.
\cite{Bayraktar} extended the work of \cite{VecerXu} and demonstrated how to implement a numerical approximation scheme for pricing Asian options in jump diffusion models.
\cite{Shiraya} proposed a new approximation formula for pricing average options under the Heston and extended SABR stochastic volatility models.
\cite{Chang} used a central $\chi$-distribution as a proxy for the true distribution to derive an approximation formula for valuing Asian options.

The remainder of the paper is organized as follows.
In Section~\ref{Sect:recursive}, we derive a recursive formula for arithmetic Asian option pricing.
In Section~\ref{Sect:application}, numerical results with the Black-Scholes framework and the exponential L\'evy model are shown.
In Section~\ref{Sect:conclusion}, we conclude the paper.

\section{Recursive formula}\label{Sect:recursive}
In this paper we introduce a complete filtered probability space $(\Omega, \mathcal F, \mathbb P)$ with a time index set $[0,T^*]$ for some fixed $T^*>0$.
We have a filtration $\{\mathcal F_t\}_{t\in [0, T^*]}$ where $\mathcal F_{T^*} = \mathcal F$.
The measure $\mathbb P$ is the physical probability measure.
All processes introduced in this paper are defined on the probability space and those processes are adapted to the filtration.

Let $S$ be an underlying asset price process.
Assume that there exists an equivalent risk-neutral measure $\Q$ under which every discounted asset price process is a martingale.
In addition, we assume that $S$ is a Markov process.
For simplicity, the instantaneous interest rate $r$ is assumed to be a constant over the time interval $[0,T^*]$.
We have $N$ observation times $0 < T_1 < \cdots < T_N = T \leq T^*$ for an arithmetic Asian option.
Let $\tau_n = T_n - T_{n-1}$, $n=1,\ldots,N$ where $T_0 = 0$.

For each $n=1,\ldots,N$ we define $T_n$-futures price $F^{(n)}_t = e^{T_n - t} S_t$ for $0 \leq t \leq T_n$, i.e.,
futures prices at time $t$ with maturity $T_n$.
Then $F^{(n)}_t$ is a $\Q$-martingale for $0 \leq t \leq T_n$.
Now we define European call and put option price functions with maturities $T_n$, $1 \leq n \leq N$.

\begin{definition}\label{Def:option}
For $T_{n-1} \leq t < T_n$, let $c_t^{(n)}(x, K)$ and $p_t^{(n)}(x, K)$ be at time $t$ European call and put option prices as functions of spot price $x$ at time $t$ and strike price $K$ with maturity $T_n$, respectively.
Subscript $t$ denotes current time, and superscript $(n)$ denotes maturity $T_n$.
In other words,
\begin{align*}
c^{(n)}_t\left( S_{t}, K \right) &= \E^{\Q}\left[\left. \e^{-r(T_n-t)}(S_{T_n} - K)^+\right| \F_{t} \right],\\
p^{(n)}_t\left( S_{t}, K \right) &= \E^{\Q}\left[\left. \e^{-r(T_n-t)} (K - S_{T_n})^+\right| \F_{t} \right].
\end{align*}
In addition, we define
$$
\phi^{(n)}_t\left( x, K \right)=\left\{
      \begin{array}{ll}
        p^{(n)}_t\left( x, K \right), & 0 < K \leq \e^{r(T_n-t)}x, \\
        c^{(n)}_t\left( x, K \right), & \e^{r(T_n-t)}x < K < \infty.
      \end{array}
    \right.
$$
For notational simplicity, when $t=T_{n-1}$, we omit subscript $t$ of $c^{(n)}_t$, $p^{(n)}_t$ and $\phi^{(n)}_t$.
In other words, $c^{(n)} = c^{(n)}_{T_{n-1}}$ and similarly for $p$ and $\phi$.
\end{definition}

For a stochastic process $X$, the quadratic variation process of $X$ is defined by
$$ [X]_t = X_t^2 - 2 \int_0^t X_u \D X_u.$$
Note that for a sequence of partition $\pi_n$ ranged over $[0,t]$, we have
$$ [X]_t =  \lim_{||\pi_n|| \rightarrow 0} \sum_{i} ( X_{t_i} - X_{t_{i-1}} )^2 \quad \textrm{in probablity.} $$
For the detailed information about quadratic variation, see~\cite{Protter}.

First, we consider an arithmetic Asian option price with a continuous underlying asset price process.
The following result is introduced in \cite{ChoeLee} and \cite{Lee}.

\begin{lemma}\label{Lemma:replication}
Let $X$ be a continuous process.
For a continuous function $g(x)$ with its anti-derivative $G(x)$, we have
\begin{align*}
\int_{u}^{t} g(X_s) \D [X]_s ={}& 2 \left( \int_{u}^{t} (G(X_u) - G(X_s)) \D X_s + \int_{X_u}^{X_t} g(K)(X_t - K) \D K \right)\\
={}& 2 \left( \int_{0}^{t} (G(X_u) - G(X_s)) \D X_s \right. \\
&+ \left.  \int_{0}^{X_u} g(K) (K-X_t)^+ \D K + \int_{X_u}^{\infty} g(K) (X_t-K)^+ \D K \right),
\end{align*}
for $0 \leq u  \leq t \leq T$.
\end{lemma}

\begin{proof}
If $f$ is twice continuously differentiable, then by It\^{o}'s lemma,
$$f(X_t) = f(X_u) + \int_{u}^{t} f'(X_s) \D X_s + \frac{1}{2}\int_{u}^{t} f''(X_s) \D [X]_s$$
and by Taylor's theorem with the integral form of the remainder term
\begin{align*}
f(X_t) ={}& f(X_u) + f'(X_u)(X_t - X_u) + \int_{X_u}^{X_t} f''(K)(X_t - K) \D K \\
={}& f(X_u) + f'(X_u)(X_t - X_u) \\
&+\int_{0}^{X_u} f''(K)(K-X_t)^+ \D K + \int^{\infty}_{X_u} f''(K)(X_t-K)^+ \D K.
\end{align*}
By comparing above equations, we have
\begin{align*}
&\int_{u}^{t} f''(X_s) \D [X]_s\\
={}& 2 \left( \int_{u}^{t} (f'(X_u) - f'(X_s)) \D X_s + \int_{X_u}^{X_t} f''(K)(X_t - K) \D K  \right) \\
={}& 2 \left( \int_{u}^{t} (f'(X_u) - f'(X_s)) \D X_s  \right. \\
&+\left.  \int_{0}^{X_u} f''(K)(K-X_t)^+ \D K + \int^{\infty}_{X_u} f''(K)(X_t-K)^+ \D K \right).
\end{align*}
Finally, substituting $g$ in place of $f''$, we have the desired result.
\end{proof}

Let $L^{2}_{\mathbb Q, [Y]}([s,t]\times\Omega)$ denote the space of adapted stochastic process $X$ such that
$$\mathbb E^{\Q} \left[ \left. \int_{s}^{t}  X^2_u \D [Y]_u  \right| \F_s \right] < \infty. \quad \textrm{a.s.}$$
Under the condition, we guarantee that the stochastic integral of $X$ with respect to a $\Q$-martingale $Y$ is a $\Q$-martingale.
For the detailed information, consult~\cite{Kuo}.

\begin{theorem}\label{Thm:representation}
For a continuous function $g(x)$ with its anti-derivative $G(x)$, if
$$G\left(F^{(n)}\right) \in L^{2}_{\mathbb Q, [F^{(n)}]}([T_{n-1},T_n]\times\Omega),$$
then for $T_{n-1} \leq t < T_n$,
$$\E^{\Q} \left[ \left. \int_{t}^{T_n} g \left( F^{(n)}_s \right) \D \left[ F^{(n)} \right]_s \right| \F_{t} \right] =
 2\e^{r(T_n -t)} \int_{0}^{\infty} g(K) \phi_t^{(n)}\left( S_{t}, K \right) \D K .$$
\end{theorem}
\begin{proof}
By Lemma~\ref{Lemma:replication},
\begin{align*}
 \int_{t}^{T_n} g \left( F^{(n)}_s \right) \D \left[ F^{(n)} \right]_s ={}& 2 \left( \int_{t}^{T_n} \left( G\left(F^{(n)}_{t} \right) - G\left( F^{(n)}_s \right) \right) \D F^{(n)}_s \right. \\
&+ \left.  \int_{0}^{F^{(n)}_t} g(K) \left( K- F^{(n)}_{T_n} \right)^+ \D K + \int_{F^{(n)}_t}^{\infty} g(K) \left( F^{(n)}_{T_n}-K \right)^+ \D K \right)\\
={}& 2 \left( \int_{t}^{T_n} \left( G\left(F^{(n)}_{t} \right) - G\left( F^{(n)}_s \right) \right) \D F^{(n)}_s \right. \\
&+ \left.  \int_{0}^{F^{(n)}_t} g(K) \left( K- S_{T_n} \right)^+ \D K + \int_{F^{(n)}_t}^{\infty} g(K) \left( S_{T_n}-K \right)^+ \D K \right).
\end{align*}
Applying $\Q$-expectation with respect to $\F_{t}$, we have the integration formula.
\end{proof}

\begin{corollary}\label{Cor:representation}
For twice continuously differentiable function $g(x)$, if
$$\frac{\D g}{\D x}\left( F^{(n)} \right) \in L^{2}_{\mathbb Q, [F^{(n)}]}([T_{n-1},T_n]\times\Omega),$$
then for $T_{n-1} \leq t < T_n$,
$$\E^{\Q} \left[ g\left( \left. F^{(n)}_{T_n} \right) \right| \F_t \right] = g \left(F^{(n)}_t \right) + e^{r(T_n-t)} \int_{0}^{\infty} \frac{\partial^2 g}{\partial x^2}(K) \phi^{(n)}_t\left( S_t, K \right) \D K$$
or equivalently,
$$\E^{\Q} \left[ g\left( \left. S_{T_n} \right) \right| \F_t \right] = g \left(e^{r(T_n-t)}S_t \right) + e^{r(T_n-t)} \int_{0}^{\infty} \frac{\partial^2 g}{\partial x^2}(K) \phi^{(n)}_t\left( S_t, K \right) \D K.$$
\end{corollary}
\begin{proof}
Applying It\^{o}'s lemma, we have
$$ g\left(  F^{(n)}_{T_n} \right) = g\left(  F^{(n)}_{t} \right) + \int_t^{T_n} \frac{\D g}{\D x}\left( F^{(n)}_s \right) \D F^{(n)}_s + \int_t^{T_n} \frac{\D^2 g}{\D x^2}\left( F^{(n)}_s \right) \D\left[ F^{(n)}\right]_s $$
and apply Theorem~\ref{Thm:representation}.
\end{proof}

Now suppose that the underlying price process is a semimartingale and hence the process can be discontinuous.
For discontinuous price processes, we have the similar results.
Instead of Lemma~\ref{Lemma:replication}, we have the following.

\begin{lemma}\label{Lemma:replication2}
Let $X^c$ be continuous part of $X$.
If $g$ is a continuous function with its anti-derivative $G$ and second anti-derivative $\bar G$, then
\begin{align*}
&\int_{u}^{t} g(X_{s-}) \D [X^c]_s + 2\sum_{u \leq s \leq t} [\Delta \bar G(X_s)  - \Delta X_s G(X_{s-})] \\
={}& 2 \left( \int_{u}^{t} (G(X_u) - G(X_s)) \D X_s + \int_{X_u}^{X_t} g(K)(X_t - K) \D K \right)\\
={}& 2 \left( \int_{u}^{t} (G(X_u) - G(X_s)) \D X_s \right.
+ \left.  \int_{0}^{X_u} g(K) (K-X_t)^+ \D K + \int_{X_u}^{\infty} g(K) (X_t-K)^+ \D K \right).
\end{align*}
\end{lemma}

\begin{proof}
If $f$ is twice continuously differentiable, then by It\^{o}'s lemma for a semmimartingale,
$$f(X_t) = f(X_u) + \int_{u}^{t} f'(X_{s-}) \D X_s + \frac{1}{2}\int_{u}^{t} f''(X_{s-}) \D [X^c]_s + \sum_{u < s \leq t}[\Delta f(X_s) - \Delta X_s f'(X_{s-})] $$
and by Taylor's theorem with the integral form of the remainder term
\begin{align*}
f(X_t) ={}& f(X_u) + f'(X_u)(X_t - X_u) + \int_{X_u}^{X_t} f''(K)(X_t - K) \D K \\
={}& f(X_u) + f'(X_u)(X_t - X_u) \\
+&\int_{0}^{X_u} f''(K)(K-X_t)^+ \D K + \int^{\infty}_{X_u} f''(K)(X_t-K)^+ \D K.
\end{align*}
By comparing above equations, we have
\begin{align*}
&\int_{u}^{t} f''(X_{s-}) \D [X^c]_s + 2\sum_{u \leq s \leq t}[\Delta f(X_s) - \Delta X_s f'(X_{s-})]\\
&=  2 \left( \int_{u}^{t} (f'(X_u) - f'(X_s)) \D X_s + \int_{X_u}^{X_t} f''(K)(X_t - K) \D K  \right) \\
&=  2 \left( \int_{u}^{t} (f'(X_u) - f'(X_s)) \D X_s  \right. +\left.  \int_{0}^{X_u} f''(K)(K-X_t)^+ \D K + \int^{\infty}_{X_u} f''(K)(X_t-K)^+ \D K \right).
\end{align*}
Finally, substituting $g$ in place of $f''$, we have the desired result.
\end{proof}

\begin{theorem}\label{Thm:representation2}
For a continuous function $g$ with its anti-derivative $G$ and second anti-derivative $\bar G$, we have
\begin{align*}
&\E^{\Q} \left[ \left. \int_{t}^{T_n} g \left( F^{(n)}_{s-} \right) \D \left[\left(F^{(n)}\right)^c\right]_s  + 2\sum_{t < s \leq T_n} [\bar G(F_s) -\bar G(F_{s-}) - \Delta F_s G(F_{s-})]\right| \F_{t}^{(n)} \right] \\
& \quad = 2e^{r(T_n-t)} \int_{0}^{\infty} g(K) \phi^{(n)}_t \left( S_t, K \right) \D K
\end{align*}
for $T_{n-1} \leq t < T_n$.
\end{theorem}

\begin{proof}
Take $\Q$-expectation to the result of Lemma~\ref{Lemma:replication2} with $X = F^{(n)}$.
\end{proof}

For a semimartingale underlying process, we have the same result with Corollary~\ref{Cor:representation}.
\begin{corollary}\label{Cor:representation2}
Suppose that the underlying stock price $S$ is a semimartingale.
For twice continuously differentiable function $g(x)$, if $g'(F^{(n)}) \in L^{2}_{\mathbb Q, [(F^c)^{(n)}]}([T_{n-1},T_n]\times\Omega)$,
then for $T_{n-1} \leq t < T_n$,
$$\E^{\Q} \left. \left[ g\left( F^{(n)}_{T_n} \right) \right| \F_t \right] = g\left( F^{(n)}_{t} \right) + \e^{r(T_n -t)} \int_{0}^{\infty} \frac{\D^2 g}{\D x^2}(K) \phi_t^{(n)}\left( S_{t}, K \right) \D K$$
or equivalently,
$$\E^{\Q} \left. \left[ g\left( S_{T_n} \right) \right| \F_t \right] = g\left( \e^{r(T_n -t)}S_t \right) + \e^{r(T_n -t)} \int_{0}^{\infty} \frac{\D^2 g}{\D x^2}(K) \phi_t^{(n)}\left(S_t, K \right) \D K.$$
\end{corollary}
\begin{proof}
Apply It\^{o}'s lemma for semimartingale and Theorem~\ref{Thm:representation2}.
\end{proof}

Now we have the following theorem for recursive formula of Asian option price with a semimartingale underlying price process.
\begin{theorem}\label{Thm:Asian_option}
For $1 \leq n \leq N-2$, define
\begin{align*}
g^{(n)}(x_1, \ldots, x_n) ={}& e^{-r\tau_{n+1}}g^{(n+1)}(x_1, \ldots, x_n, e^{r\tau_{n+1}}x_n)\\
&+ \int_0^{\infty}\frac{\partial^2 g^{(n+1)}}{\partial x_{n+1}^2}(x_1, \ldots, x_n, K) \phi^{(n+1)}(x_n, K) \D K
\end{align*}
and
\begin{align*}
g^{(N-1)}(x_1, \ldots, x_{N-1}) ={}& \frac{1}{N}\left(c^{(N)}
\left(x_{N-1} , NE-\sum_{i=1}^{N-1}x_i \right)1_{\left\{ \sum_{i=1}^{N-1}x_i < NE \right\}}  \right.\\
&\left. + e^{-r\tau_N}\left(\sum_{i=1}^{N-1}x_i + x_{N-1} -NE \right)1_{\left\{\sum_{i=1}^{N-1}x_i \geq NE  \right\}}\right).
\end{align*}
Assume that $g^{(n)}$ is continuously twice differentiable with respect to $x_n$ and
$$\frac{\partial^2 g^{(n)}}{\partial x_{n}^2} \left(S_{T_1}, \ldots, S_{T_{N-1}}, F^{(n)}\right) \in L^{2}_{\mathbb Q, [F^{(n)}]}([T_{n-1},T_n]\times\Omega)$$
for $1\leq n \leq N-1$.
Then the discounted risk-neutral expectation of Asian option payoff at time $0\leq t< T_1$ is given by
\begin{equation*}
\E^{\Q}_t \left[  e^{-r(T-t)}\left(\frac{1}{N}\sum_{i=1}^{N} S_{T_i} - E \right) \right] =
e^{-r(T_1-t)} g^{(1)}\left(e^{r(T_1-t)}S_t\right) + \int_0^\infty \frac{\D^2 g^{(1)}}{\D x_1^2}(K) \phi_t^{(1)}\left( S
_t, K\right) \D K.
\end{equation*}
\end{theorem}

\begin{proof}
Note that
\begin{align*}
&\E^{\Q} \left[e^{-r\tau_N} \left. \left(\frac{\sum_{i=1}^{N} S_{T_i}}{N} - E \right)^+ \right| \F_{T_{N-1}} \right]\\
&= \frac{1}{N} \left( c^{(N)}\left( S_{T_{N-1}}, NE - \sum_{i=1}^{N-1} S_{T_i} \right) 1_{\{\sum_{i=1}^{N-1}S_{T_i} < NE\}} \right. \\
&\quad + \left.e^{-r\tau_N}\left(\sum_{i=1}^{N-1}S_{T_i} + S_{T_{N-1}} -NE \right) 1_{\{\sum_{i=1}^{N-1}S_{T_i} \geq NE\}} \right)\\
&= \frac{1}{N} \left(c^{(N)}\left( e^{r\tau_{N}}S_{T_{N-1}}, NE - \sum_{i=1}^{N-1} S_{T_i} \right) 1_{\{\sum_{i=1}^{N-1}S_{T_i} < NE\}} \right.\\
&\quad+ \left.e^{-r\tau_N}\left(\sum_{i=1}^{N-1}S_{T_i} + S_{T_{N-1}} -NE \right) 1_{\{\sum_{i=1}^{N-1}S_{T_i} \geq NE\}} \right)\\
&= g^{(N-1)}\left( S_{T_1}, \ldots, S_{T_{N-1}} \right).
\end{align*}
Furthermore, by Corollary~\ref{Cor:representation2}, we have
\begin{align*}
&\E^{\Q}\left[ e^{-r\tau_{N-1}}\left. g^{(N-1)}\left( S_{T_1}, \ldots, S_{T_{N-1}} \right) \right| \F_{T_{N-2}}\right]\\
&=e^{-r\tau_{N-1}}g^{(N-1)}\left( S_{T_1}, \ldots, S_{T_{N-2}}, e^{r\tau_{N-1}}S_{T_{N-2}} \right)\\
&+ \int_{0}^{\infty} \frac{\partial^2 g^{(N-1)}}{\partial x^2_{N-1}}\left(S_{T_1}, \ldots, S_{T_{N-2}},K \right) \phi^{(N-1)}\left( S_{T_{N-2}}, K \right) \D K.
\end{align*}
Applying the result recursively, we have the desired result.
\end{proof}

One of the advantages to have a solution of option price is that we are able to compute the sensitivities easily.
\begin{corollary}\label{Cor:delta}
The delta of an Asian option at time $0\leq t< T_1$ is
$$\Delta_t = \frac{\D g^{(1)}}{\D x_1}\left( e^{r(T_1-t)}S_t \right) + \e^{r(T_1-t)}\int_0^\infty \frac{\D^2 g^{(1)}}{\D x_1^2}(K) \frac{\partial \phi_t^{(1)}}{\partial x} \left( S_t, K \right)\D K.$$
\end{corollary}
\begin{proof}
By differentiating option price with respect to underlying price, we have the desired result.
\end{proof}

\begin{remark}\label{Remark:relation}
In Theorem~\ref{Thm:Asian_option}, the function $g^{n}(x_1, \ldots, x_n)$ denotes the time $T_n$ arithmetic Asian option price
with given observed underlying prices $x_1, \ldots, x_n$ at $T_1, \ldots, T_n$, respectively.
For $T_n \leq t < T_{n+1}$, the Asian call option price at time $t$ is represented as another Asian option price. That is
$$ \E^{\Q} \left. \left[ \e^{-r(T_N - t)} \left( \frac{1}{N}\sum_{i=1}^{N} S_{T_i} - E\right)^+ \right| \F_t \right]
=  \frac{N-n}{N} \E^{\Q} \left. \left[  \e^{-r(T_N - t)}\left( \frac{1}{N-n}{\sum_{i=n+1}^{N} S_{T_i}} - E_n \right)^+ \right| \F_t \right] $$
where
$$ E_n = \frac{NE - \sum_{i=1}^{n} S_{T_i}}{N-n} .$$
\end{remark}

\begin{example}
Consider an arithmetic Asian option under the Black-Scholes framework with two observation times.
Under the assumption the arithmetic option price is represented by an integration formula.
We have
$$ g^{(1)}(x) = \frac{1}{2}c^{(2)}(x, 2E-x)1_{\{x < 2E\}} + \e^{-r\tau_2}(x-E)1_{\{x \geq 2E\}}.$$
and
$$ \frac{\D g^{(1)}}{\D x}(x) = \frac{1}{2}\left(\frac{\partial c^{(2)}}{\partial x}(x, 2E-x)  - \frac{\partial c^{(2)}}{\partial K}(x, 2E-x)  \right) 1_{\{x < 2E\}} + \e^{-r\tau_2}1_{\{x \geq 2E\}}$$
and
\begin{align*}
\frac{\D^2 g^{(1)}}{\D x^2}(x) ={} \frac{1}{2}\left( \frac{\partial^2 c^{(2)}}{\partial x^2}(x, 2E-x) + \frac{\partial^2 c^{(2)}}{\partial K^2}(x, 2E-x) \right. \left.- 2 \frac{\partial^2 {c^{(2)}}}{\partial x \partial K}(x, 2E-x)\right)1_{\{x < 2E\}}.
\end{align*}
Note that under the framework the European call option price is
$$
c^{(2)}(x, K) =  x  N(d_1(x,K)) - K\e^{-r\tau_2} N(d_2(x,K))
$$
where $N$ is the standard normal c.d.f. and
$$ d_1(x,K) = \frac{\log \left(\frac{x}{K} \right) + \left( r + \frac{\sigma^2}{2}\right) \tau_2}{\sigma \sqrt{\tau_2}}, \quad d_2(x,K) = d_1(x,K) - \sigma\sqrt{\tau_2}.$$
Since
\begin{align*}
\frac{\partial c^{(2)}}{\partial x}(x,K) &=  N(d_1(x,K)),\\
\frac{\partial c^{(2)}}{\partial K}(x,K) &= -\e^{-r\tau_2} N(d_2(x,K)),\\
\frac{\partial^2 c^{(2)}}{\partial x^2}(x,K) &=  \frac{N'(d_1(x,K))}{x \sigma \sqrt{\tau_2}} = \frac{1}{x\sigma\sqrt{2\pi\tau_2}}\exp\left(-\frac{d_1^2(x,K)}{2}\right),\\
\frac{\partial^2 c^{(2)}}{\partial K^2}(x,K) &= \e^{-r\tau_2} \frac{N'(d_2(x,K))}{K \sigma \sqrt{\tau_2}} = \frac{\e^{-r\tau_2}}{K\sigma\sqrt{2\pi\tau_2}}\exp\left(-\frac{d_2^2(x,K)}{2}\right),\\
\frac{\partial^2 c^{(2)}}{\partial x \partial K}(x,K) &= -\e^{-r\tau_2} \frac{N'(d_1(x,K))}{K\sigma\sqrt{\tau_2}}= - \frac{\e^{-r\tau_2}}{K\sigma\sqrt{2\pi\tau_2}}\exp\left(-\frac{d_1^2(x,K)}{2}\right),
\end{align*}
we have
\begin{align*}
\frac{\D^2 g^{(1)}}{\D x^2}(x) ={}& \frac{1}{2\sigma \sqrt{2\pi\tau_2}} \left( \frac{1}{x} \exp\left( - \frac{d_1^2(x, 2E-x)}{2} \right) + \frac{e^{-r\tau_2}}{2E-x}\exp\left(-\frac{d_2^2(x, 2E-x)}{2}\right) \right.\\
&+\left. \frac{2e^{-r\tau_2}}{2E-x} \exp\left( -\frac{d_1^2( x, 2E-x)}{2} \right) \right)1_{\{x < 2E\}}.
\end{align*}
Therefore, the Asian option price at time $0 \leq t < T_1$ is
\begin{align*}
&\E^{\Q} \left[ \left. \e^{-r(T_2-t)} \left( \frac{S_{T_1} + S_{T_2}}{2} - E \right)^+  \right| \F_t \right] \\
={} & \e^{-r(T_1-t)}\left( \frac{1}{2}c^{(2)}\left( \e^{r(T_1-t)}S_t, 2E-e^{r(T_1-t)}S_t \right)1_{\{e^{r(T_1-t)}S_t < 2E\}} \right.\\
& +\left. \e^{-r\tau_2}\left( e^{r(T_1-t)}S_t-E \right)1_{\{e^{r(T_1-t)}S_t \geq 2E\}} \right)\\
& + \frac{1}{2\sigma \sqrt{2\pi\tau_2}}\int_0^{2E} \left\{\frac{1}{K} \exp\left(-\frac{d_1^2(K, 2E-K)}{2} \right) + \frac{\e^{-r\tau_2}}{2E-K}\exp\left(-\frac{d_2^2(K, 2E-K)}{2}\right) \right. \\
& \left.+ \frac{2e^{-r\tau_2}}{2E-K} \exp\left( -\frac{d_1^2(K, 2E-K)}{2} \right) \right\}\phi^{(1)}_t\left(S_t, K\right) \D K.
\end{align*}

\end{example}

To deal with arithmetic Asian option prices with numerous observation times
it is better to use the next result.
For simplicity, in the next theorem assume that $\tau = \tau_n$ for all $1 \leq n \leq N$, i.e., equally distributed observation times.
\begin{theorem}\label{Thm:basic}
Let $\bar A^{(\ell)}(w)$, $-\infty < w < \infty$, denote an arithmetic Asian option price with $\ell$-observation times with current spot price 1 and strike price $w$.
Then
\begin{align*}
\bar A^{(\ell)} \left( w \right)={}&
\frac{(\ell-1)}{\ell} \bar A^{(\ell-1)}\left(\frac{w\ell-e^{r\tau}}{e^{r\tau}(\ell-1)}\right)  \\
&+ \int_{0}^{\infty}  \frac{w^2\ell}{K^3(\ell-1)}\frac{\partial^2 \bar A^{(\ell-1)}}{\partial w^2} \left(\frac{w\ell-K}{K(\ell-1)}\right)\phi(1, K) \D K
\end{align*}
where
$\bar A^{(1)} \left( w \right) =  c(1,w)1_{\{w>0\}} + (1 - \e^{-r \tau}w)1_{\{w \leq 0\}}$.
\end{theorem}

\begin{proof}
We can rewrite $g^{(n)}$ as a two dimensional function of $(\sum_{i=1}^{n-1} x_i, x_n)$.
By Remark~\ref{Remark:relation},
$$ g^{(n)}\left(\sum_{i=1}^{n-1} x_i, x_n\right)  = \frac{x_n(N-n)}{N} \bar A^{(N-n)} \left( \frac{NE - \sum_{i=1}^{n}x_i}{x_n(N-n)} \right).$$
Note that
$$ \frac{\partial^2 g^{(n)}}{\partial x_n^2}(u,x_n) =
\frac{(NE-u)^2}{x_n^3N(N-n)}\frac{\partial^2 \bar A^{(N-n)}}{\partial w^2} \left(\frac{NE-u-x_n}{x_n(N-n)}\right)$$
where $u = \sum_{i=1}^{n-1} x_i$.
By the definition of $g^{(n)}$,
\begin{align*}
\frac{x_n(N-n)}{N} \bar A^{(N-n)} \left( \frac{NE - \sum_{i=1}^{n}x_i}{x_n(N-n)} \right)=
\frac{x_n(N-n-1)}{N} \bar A^{(N-n-1)}\left(\frac{NE-\sum_{i=1}^{n}x_i - e^{r\tau}x_n}{\e^{r\tau}x_n(N-n-1)}\right)  \\
+ \int_{0}^{\infty}  \frac{(NE -\sum_{i=1}^{n}x_i)^2}{K^3N(N-n-1)}\frac{\partial^2 \bar A^{(N-n-1)}}{\partial w^2} \left(\frac{NE-\sum_{i=1}^{n}x_i-K}{K(N-n-1)}\right)\phi(x_n, K) \D K.
\end{align*}
Setting $x_n =1$, we have
\begin{align*}
\bar A^{(N-n)} \left( \frac{NE - \sum_{i=1}^{n}x_i}{(N-n)} \right)=
\frac{N-n-1}{N-n} \bar A^{(N-n-1)}\left(\frac{NE-\sum_{i=1}^{n}x_i - \e^{r\tau}}{\e^{r\tau}(N-n-1)}\right)  \\
+ \int_{0}^{\infty}  \frac{(NE -\sum_{i=1}^{n}x_i)^2}{K^3(N-n)(N-n-1)}\frac{\partial^2 \bar A^{(N-n-1)}}{\partial w^2} \left(\frac{NE-\sum_{i=1}^{n}x_i-K}{K(N-n-1)}\right)\phi(1, K) \D K.
\end{align*}
Put
$$w =  \frac{NE - \sum_{i=1}^{n}x_i}{N-n}.$$
Then
\begin{align*}
\bar A^{(N-n)} \left( w \right)={}&\frac{(N-n-1)}{(N-n)} \bar A^{(N-n-1)}\left(\frac{w(N-n)-\e^{r\tau}}{\e^{r\tau}(N-n-1)}\right)  \\
&+ \int_{0}^{\infty}  \frac{w^2(N-n)}{K^3(N-n-1)}\frac{\partial^2 \bar A^{(N-n-1)}}{\partial w^2} \left(\frac{w(N-n)-K}{K(N-n-1)}\right)\phi(1, K) \D K
\end{align*}
and by setting $\ell = N-n$, we complete the proof.
\end{proof}

\section{Numerical results}\label{Sect:application}
\subsection{The Black-Scholes model}
In this subsection we consider arithmetic Asian option prices with geometric Brownian motions.
For the numerical work, we apply Theorem~\ref{Thm:basic}.
First we compute $\bar A^{(1)} (w)$ over reasonable strike horizon, for example $w \in (0,2)$, with step size 0.0025 and calculate its second derivatives.
The second derivatives are approximated by differences of the original function values.
Next we compute $\bar A^{(2)} (w)$ over the strike horizon with numerical integration formula in Theorem~\ref{Thm:basic}.
For the numerical integration the strike horizon is set from 0.01 to 2 with step size 0.001 and Simpson's rule is applied.
We repeat the same process until to get $\bar A^{(90)} (w)$ and scale the option value based on current spot price.
The numerical result is shown in Table~\ref{Table:BS_AsianN} where $S_0 = 100, \sigma = 0.2, r=0.05$, $\tau = 1$ day and $T_N = 90$ days.

We report the Monte Carlo simulation results with $2\times10^6$ paths as a comparison.
We also report the results based on the PDE methods of \cite{Vecer2001}.
Mathematica implementation of the procedure which comes from the author's homepage is used and the boundaries are set to $-5 \sigma T_N$ and $5\sigma T_N$.
\begin{table}
\centering
\caption{Asian option prices under Black-Scholes framework with 90 observation times where $S_0 = 100, \sigma = 0.2, r=0.05$, $\tau = 1$ day and $T_N = 90$ days.}\label{Table:BS_AsianN}
\begin{tabular}{cccc}
\hline
Asian strike  & our method & M.C. & Vecer\\
\hline
80 & 20.3732 &  20.3728 & 20.3783\\
85 & 15.4368 &  15.4364 & 15.4426\\
90 & 10.5501 &  10.5481 & 10.5539\\
95 &  6.0270  &  6.0206 & 6.0214\\
100 & 2.6157 &  2.6082 & 2.6020\\
105 & 0.8042 &  0.7975 & 0.7914\\
110 & 0.1709 &  0.1674 & 0.1650\\
115 & 0.0252 &  0.0243 & 0.0238\\
120 & 0.0026 &  0.0024 & 0.0025\\
\hline
\end{tabular}
\end{table}

\subsection{Exponential L\'evy model}
Previously, we compute Asian option prices under the Black-Scholes framework.
For more general cases including exponential L\'evy models, our method to compute arithmetic Asian option price is applicable.
Once we have analytic solution for the European option price as a function of underlying asset price,
the computation of an Asian option price is straightforward.

It is well known that the price of a European call option with maturity $\tau$ and strike $K$ is represented by
$$ c(S_0,K) =  S_0 \Pi_1(F_0,K) - Ke^{-r\tau} \Pi_2(S_0,K)$$
where $P_1$ and $P_2$ are probabilities satisfying
\begin{align*}
\Pi_1(S_0,K) &= \frac{1}{2} + \frac{1}{\pi} \int_{0}^{\infty} \mathrm{Re} \left[\frac{e^{-iu\log K} \psi_\tau(u-i)}{iu\psi_\tau(-i)} \right] \D u, \\
\Pi_2(S_0,K) &= \frac{1}{2} + \frac{1}{\pi} \int_{0}^{\infty} \mathrm{Re} \left[\frac{e^{-iu\log K} \psi_\tau(u)}{iu} \right] \D u.
\end{align*}
and $\psi_\tau(u) = \E^{\Q} [\exp(iu \log S_\tau)]$.

However, it is difficult to evaluate the numerical integrals directly since the integrands diverge as $u \rightarrow 0$.
To apply numerical evaluation, we may use the analytic expressions for Fourier transforms of
some kinds of modified call option prices
such as dampened option price, the time value of option price or the option price subtracted by Black-Scholes price.
The reason to consider these modified call option price function is to obtain a square-integrable function.
With square-integrable functions we apply Fourier transform method.
These approaches are introduced and explained in \cite{CarrMadan} and \cite{ContTankov}.

No matter what kind of modified option price function is chosen, the algorithms to compute an Asian option price are similar.
First, define a modified European call option price $z_\tau$ as a function of log strike $k$.
For example, we choose the dampened option price as a modified price function for the numerical work.
By setting $S_0 = 1$, we define a modified option price $z_\tau$ by
$$ z_\tau(k) = \e^{\alpha k}  \E^{\Q}[\e^{-r\tau}( S_\tau - \e^k)^+].$$

Second, we derive the analytic formula of the inverse Fourier transform of $z_\tau(k)$.
Let
$$\zeta_\tau(v) = \int_{-\infty}^{\infty} \e^{ivk} z_\tau(k) \D k.$$
Then
\begin{align*}
\zeta_\tau(v) &= \int_{-\infty}^{\infty} \e^{ivk} z_\tau(k) \D k\\
&= \frac{e^{-r\tau}\psi_\tau(v-(\alpha+1)i)}{\alpha^2+\alpha-v^2+i(2\alpha+1)v}.
\end{align*}
Then the call option price $\bar c_\tau(k)$ with $S_0 =1$ and log strike $k$ is obtained by the multiplication of Fourier transform of $\zeta_\tau$ and $\e^{-\alpha k}$.
More precisely,
\begin{align*}
\bar c_\tau(k) &= \frac{\e^{-\alpha k}}{2\pi} \int_{-\infty}^{\infty} \e^{-ivk} \zeta_\tau(v) \D v \\
&= \frac{\e^{-\alpha k}}{\pi} \int_{0}^{\infty}  \e^{-ivk} \zeta_\tau(v)  \D v  .
\end{align*}
An efficient way to compute the Fourier transform part is based on FFT.
For the details of how to apply FFT method to compute European option prices are explained in \cite{CarrMadan}.
The European call option price with current price $x$ and strike $K$ is obtained by
$$ c_{\tau}(x,K) = x \bar c_{\tau}\left(\log \frac{K}{x} \right).$$

We can compute the approximate derivatives $\bar c_{\tau}'$ and $\bar c_{\tau}''$ by the differences of obtained discrete values of $\bar c_{\tau}$.
Also we can compute derivatives by FFT method
\begin{align*}
\frac{\D \bar c_{\tau}}{\D k}(k) &= \frac{\e^{-\alpha k}}{2\pi} \int_{-\infty}^{\infty} -iv \e^{-ivk} \zeta_{\tau}(v) \D v\\
\frac{\D^2 \bar c_{\tau}}{\D k^2}(k) &= \frac{\e^{-\alpha k}}{2\pi} \int_{-\infty}^{\infty} v^2 \e^{-ivk} \zeta_{\tau}(v) \D v .
\end{align*}
For the derivatives based on FFT, see \cite{Johnson}.
It turns out that the Asian option prices based on two methods are very similar.
In the numerical work we simply use the numerical derivative computed by differences of the obtained values of $\bar c_{\tau}$.

Finally, we compute the Asian option price by Theorem~\ref{Thm:basic} using the same method explained in the previous subsection.

For underlying price process, we assume an exponential variance Gamma process.
Let $\gamma_t$ be a Gamma process with mean rate parameter 1 and variance parameter $\nu$.
Consider a variance Gamma process given by
$$X_t = \theta t + \sigma W_{\gamma_t}$$
for a standard Brownian motion $W$.
Assume that the risk-neutral underlying price process is
$$ S_t = S_0 \exp( rt + X_t(\sigma, \theta, \nu) + \omega t)$$
where $\omega = (1/\nu)\log(1-\theta \nu -\frac{1}{2}\sigma^2 \nu)$.
\cite{MadanCarrChang} show that the characteristic function of $\log S_T$ with $S_0 =1$ is
$$\psi_\tau(u) = \exp \{iu(r +\omega)\tau\}\left( 1-i\theta\nu u + \frac{1}{2} \sigma^2 u^2 \nu \right)^{-\tau/\nu}.$$

For FFT, the number of discrete summation is set to $2^{14}$ and the effective upper limit is set to 2000.
We compute $\bar A^{(1)} (w)$, $w \in (0,2)$, with step size 0.005 and compute its second derivatives as in the previous subsection.
For the numerical integration the strike horizon is set from 0.1 to 2 with step size 0.001 and Simpson's rule is applied.

In Table~\ref{Table:VG_AsianN} the numerical results of arithmetic Asian option pricing are shown where $\sigma=0.3$, $\nu=0.3$, $\theta=-0.1$, $r=0.05$, $\tau=$1 day and $T_N=90$ days.
We report the Monte Carlo simulation results with $2\times10^6$ paths as a comparison.

\begin{table}
\centering
\caption{Asian option prices under VG process with 90 observation times with $\sigma=0.3$, $\nu=0.3$, $\theta=-0.1$, $r=0.05$,$\tau=$1 day and $T_N=90$ days}\label{Table:VG_AsianN}
\begin{tabular}{ccc}
\hline
Asian strike  & Our method & M.C.\\
\hline
80 & 20.4850 & 20.4789\\
85 & 15.6820 & 15.6748\\
90 & 11.0325 & 11.0223\\
95 & 6.7146  & 6.6981\\
100 & 3.1644 & 3.1382\\
105 & 1.3803 & 1.3558\\
110 & 0.6958 & 0.6812\\
115 & 0.3794 & 0.3714\\
120 & 0.2185 & 0.2138\\
\hline
\end{tabular}
\end{table}

\section{Concluding remark}\label{Sect:conclusion}
We derive a recursive formula for arithmetic Asian option prices such that they are consistent with European option prices.
Based on a quadratic variation method, Asian option prices are represented by an integration formula whose integrand depends on European option prices.
Our method is applicable for use with a semimartingale underlying price process.
Applications with the Black-Scholes and exponential variance Gamma option models are shown.
Because the European option prices under Black-Scholes and exponential variance Gamma models are known,
we are able to compute the arithmetic Asian option prices under the given frameworks.
As long as the European option prices are known, our method is applicable.
The accuracy of the Asian option price based on our method depends on whether the European option price is correct.

\end{document}